\documentclass[a4paper,11pt]{amsart}
\usepackage{amsmath}
\usepackage{amssymb}
  \usepackage{paralist}
  \usepackage{graphics} 
  \usepackage{epsfig} 
\usepackage{graphicx}  \usepackage{epstopdf}
 \usepackage[colorlinks=true]{hyperref}
\hypersetup{urlcolor=blue, citecolor=red}

\newtheorem{thmintro}{Theorem}

\newtheorem{theorem}{Theorem}[section]
\newtheorem*{theorem*}{Theorem}
\newtheorem{corollary}{Corollary}

\newtheorem{lemma}[theorem]{Lemma}

\newtheorem*{question}{Question}
\theoremstyle{definition}

\newtheorem{remark}{Remark}
\newtheorem{example}{Example}

\newcommand{\C} {\ensuremath{\mathbb{C}}}
\newcommand{\R} {\ensuremath{\mathbb{R}}}
\newcommand{\N} {\ensuremath{{\rm N}}}

\newcommand{\Z} {\ensuremath{\mathbb{Z}}}
\renewcommand{\P} {\ensuremath{\mathbb{P}}}
\newcommand{\F} {\ensuremath{\mathbb{F}}}
\newcommand{\calC}{\mathcal{C}}
\newcommand{\calD}{\mathcal{D}}
\newcommand{\GL} {{\rm GL}}
\newcommand{\PGL} {{\rm PGL}}
\newcommand{\barS}{\bar S}
\newcommand{\RM} {\mathcal{R}\mathcal{M}}
\newcommand{\PRM} {\mathcal{P}\mathcal{R}\mathcal{M}}

\DeclareMathOperator{\supp}{supp}
\DeclareMathOperator{\id}{id}
\DeclareMathOperator{\weight}{wt}
\DeclareMathOperator{\Stab}{Stab}
\DeclareMathOperator{\Tr}{Tr}

\title[Symmetries of weight enumerators] 
{Symmetries of weight enumerators and applications to Reed-Muller codes}

\author[M. Borello and O. Mila]{Martino Borello and Olivier Mila}

\begin{document}
\maketitle


\begin{abstract}
Gleason's 1970 theorem on weight enumerators of self-dual codes has played a crucial role for research in coding theory during the last four decades. Plenty of generalizations have been proved but, to our knowledge, they are all based on the symmetries given by MacWilliams' identities. This paper is intended to be a first step towards a more general investigation of symmetries of weight enumerators. We list the possible groups of symmetries, dealing both with the finite and infinite case, we develop a new algorithm to compute the group of symmetries of a given weight enumerator and apply these methods to the family of Reed-Muller codes, giving, in the binary case, an analogue of Gleason's theorem for all parameters.
\end{abstract}

\tableofcontents

\section{Introduction}\label{sec:intro}
Gleason's 1970 theorem about the weight enumerators of self-dual codes is, as Sloane puts it, ``one of the most remarkable theorems in coding theory'' \cite{Slo06}. Weight enumerators of self-dual doubly-even codes have a group of symmetry containing a subgroup of order 192, that is generated by the symmetry coming from MacWilliams' identities related to the self-duality condition and by the symmetry coming from the divisibility condition. This simple observation led Gleason to prove, by classical arguments from invariant theory, that the weight enumerator of a self-dual doubly-even code belongs to the polynomial ring $\C[w_{\hat{\mathcal{H}}_3}(x,y),w_{\mathcal{G}_{24}}(x,y)]$, where $w_{\hat{\mathcal{H}}_3}(x,y)$ and $w_{\mathcal{G}_{24}}(x,y)$ are the weight enumerators of the extended Hamming code of length 8 and of the extended binary Golay code of length 24 respectively \cite{Gl}.
The importance of Gleason's Theorem is surely due to its fecundity and to the numerous new research problems it generated. For example, it implies that self-dual doubly-even codes exist only for lengths multiple of 8.
Moreover, Mallows and Sloane derived upper bounds on the minimum distance of such codes using Gleason's Theorem, leading to the notion of extremal codes (those which attain the bound) \cite{MS73}.
Finally, and most remarkably, weight enumerators of extremal self-dual doubly-even codes can be determined starting from Gleason's theorem: for some lengths the corresponding codes are known, and it is useful
to have their weight enumerators on record, in the other cases it is hoped that knowledge of the weight enumerator will assist in deciding the existence of the codes.
The long-standing open problem on the existence of an extremal self-dual doubly-even code of length 72 \cite{S73} is probably one of the main examples of the active research that has arisen from Gleason's theorem (see \cite{B15} for a summary of the some of the most recent results about it).

Many generalizations of Gleason's theorem to other family of self-dual codes have been proved, as documented in \cite{NRS06}. All of them make use of MacWilliams' identities and their generalizations, which give a symmetry of the weight enumerator only if the code is self-dual or eventually formally self-dual. To our knowledge, there is no systematic research on other cases, that is codes for which MacWilliams' identities do not give a symmetry. However, many interesting families of codes (e.g. Reed-Muller codes) do not have this property, and it would be useful to have a similar result about their weight enumerators. This paper is intended to be a first step in this direction.

The following questions guided our work:
\begin{enumerate}
    \item Which are the possible groups of symmetries of a weight enumerator?
    \item If we know a weight enumerator of a code, how can we compute efficiently its symmetries?
    \item Once that we have computed some symmetries of the weight enumerator of a code, can we prove that they are symmetries of other weight enumerators of codes belonging to the same family (as in the case of self-dual doubly-even codes)?
    \item Can we determine, with these methods, new properties or unknown weight enumerators?
\end{enumerate}

Note that answering these questions is in general quite difficult, since the nature of weight enumerator is essentially combinatorial while the codes are geometric objects.
It is hence not evident how and if properties of codes will give rise to symmetries of their weight enumerators.

About the first two questions, a partial answer is given in \cite{BO00} for general polynomials. In that paper the authors address the problem of finding symmetries of homogeneous two-variables polynomials and they develop an algorithm based on moving frames to compute their group of symmetries. In the paper this algorithm is implemented in {\sc Maple}, but as the authors say, this computer algebra system has the weakness of a poor handling with both algebraic numbers and rational algebraic functions. We use some of their ideas to develop a new method which can be easily implemented in {\sc Magma} \cite{Magma} and we concentrate on the case of weight enumerators.

A homogeneous polynomial $p(x,y)\in \C[x,y]$ defines a variety $V(p(x,y))$ in the projective line $\mathbb{P}^1(\mathbb{C})$ (the set of ``roots'' of $p(x,y)$). The classical
action of ${\rm PGL}_{2}(\mathbb{C})$ on the projective line
$\mathbb{P}^1(\mathbb{C})$ (which is sharply $3$-transitive) induces an action of the projection of the group of symmetries of $p(x,y)$ in ${\rm PGL}_{2}(\mathbb{C})$ on $V(p(x,y))$.
We  will  use  this  simple  observation  later  on  to  prove  the following result.

\begin{thmintro} \label{thm:finite-roots}The group of symmetries of the weight enumerator $p(x,y)$ of a code is finite if and only $\#V(p(x,y))\geq 3$.
\end{thmintro}

If $\#V(p(x,y))\geq 3$, the projection of the group of symmetries of the weight enumerator is a finite subgroup of ${\rm PGL}_{2}(\mathbb{C})$, and these subgroups are classified by Blichfeldt in \cite{B17} (see Theorem \ref{thm:Bli}).
Consequently, the group of symmetries is a central extension of one of these subgroups, and it is well-known that isomorphism classes of central extensions are counted by second cohomology groups. On the other hand, we provide an almost complete classification of linear codes for which $\#V(p(x,y))<3$ (see Theorem \ref{thm:2roots}).  This answers completely the question about the possible groups of symmetry. Moreover, the sharply $3$-transitive action of ${\rm PGL}_{2}(\mathbb{C})$ on
$\mathbb{P}^1(\mathbb{C})$ provides a tool to develop an efficient algorithm to compute symmetries of a given weight enumerator.

In order to give an answer to question (3), one has to choose a family to study. In this paper we deal with the family of Reed Muller codes (affine or projective), for their importance in coding theory and their relation with algebraic geometry.
Determining their weight enumerator is a difficult task that is related to the counting of $\mathbb{F}_q$-rational points on hypersurfaces. This connection between two difficult problems may lead to partial solutions, as shown in \cite{El} and \cite{Na}.
Along those lines, we contribute to the investigation of weight enumerators of Reed-Muller codes by studying their group of symmetries and the related invariant ring.
In the case of binary Reed-Muller codes we have the complete picture, while in the other cases we present only partial results.

\begin{thmintro} \label{thm:reed-muller-sym}
Let $\mathcal{C}:=\mathcal{RM}(r,m)$,
$i:=\lfloor\frac{m-1}{r}\rfloor$ if $r\neq 0$ and $i:=m$ otherwise,
and $j:=\lfloor\frac{m-1}{m-r-1}\rfloor$ if $r\neq m-1$ and $j:=m$
otherwise. Denote $\zeta'_n$ a primitive $n$-th root of unity and
$\barS$ the projection of the group of symmetries of the weight
enumerator $w_{\mathcal{C}}(x,y)$ of $\mathcal{C}$ in ${\rm
PGL}_{2}(\mathbb{C})$.  Then
\begin{enumerate}
    \item \label{itB:mleqr}if $m\leq r$, then $w_{\mathcal{C}}(x,y)=(x+y)^{2^m};$
    \item \label{itB:r<m<2r+1,jneq2}if $r<m<2r+1$ and $j\neq 2$, then $\barS=\left\langle
    \left[\begin{smallmatrix}1&0\\0 & -1\end{smallmatrix}\right],
    \left[\begin{smallmatrix}1+\zeta'_{2^j}&1-\zeta'_{2^j}\\1-\zeta'_{2^j} & 1+\zeta'_{2^j}\end{smallmatrix}\right]
    \right\rangle\cong D_{2^j}
    $ and
    \[
    w_{\mathcal{C}}(x,y)\in \C[w_{\mathcal{RM}(j-1,j)}(x,y),w_{\mathcal{RM}(j-1,j+1)}(x,y)];
    \]
    \item \label{itB:r<m<2r+1,j=2}if $r<m<2r+1$ and $j=2$, then $\barS\supseteq\left\langle
    \left[\begin{smallmatrix}1&0\\0 & -1\end{smallmatrix}\right],
    \left[\begin{smallmatrix}1+\zeta'_{4}&1-\zeta'_{4}\\1-\zeta'_{4} & 1+\zeta'_{4}\end{smallmatrix}\right]
    \right\rangle\cong D_{4}
    $,
    \[
    w_{\mathcal{C}}(x,y)\in \C[w_{\mathcal{RM}(1,2)}(x,y),w_{\mathcal{RM}(1,3)}(x,y)]
    \]
    and, if $\barS\neq D_4$, then
    \[
      w_{\mathcal{C}}\left(\frac{-x+((\zeta_8')^3+(\zeta_8')^2-\zeta_8')y}{2},\frac{((\zeta_8')^3-(\zeta_8')^2-\zeta_8')x+y}{2}\right)=\pm w_{\mathcal{C}}(x,y);
    \]
    \item \label{itB:m=2r+1}if $m=2r+1$, then $\barS=\left\langle
    \left[\begin{smallmatrix}1&0\\0 & \zeta'_4\end{smallmatrix}\right],
    \frac{1}{\sqrt{2}}\left[\begin{smallmatrix}1&1\\1 & -1\end{smallmatrix}\right]
    \right\rangle\cong S_4
    $ and
    \[
    w_{\mathcal{C}}(x,y)\in \C[w_{\hat{\mathcal{H}}_3}(x,y),w_{\mathcal{G}_{24}}(x,y)];
    \]
    \item \label{itB:m>2r+1,ineq2}if $m>2r+1$ and $i\neq 2$, then $\barS=\left\langle
    \left[\begin{smallmatrix}1&0\\0 & \zeta'_{2^i}\end{smallmatrix}\right],
    \left[\begin{smallmatrix}0&1\\1 & 0\end{smallmatrix}\right]
    \right\rangle\cong D_{2^i}
    $ and
    \[
    w_{\mathcal{C}}(x,y)\in \C[w_{\mathcal{RM}(0,i)}(x,y),w_{\mathcal{RM}(1,i+1)}(x,y)];
    \]
    \item \label{itB:m=2r+1,i=2}if $m>2r+1$ and $i=2$, then $\barS\supseteq \left\langle
    \left[\begin{smallmatrix}1&0\\0 & \zeta'_{4}\end{smallmatrix}\right],
    \left[\begin{smallmatrix}0&1\\1 & 0\end{smallmatrix}\right]
    \right\rangle\cong D_{4}
    $,
    \[
    w_{\mathcal{C}}(x,y)\in \C[w_{\mathcal{RM}(0,2)}(x,y),w_{\mathcal{RM}(1,3)}(x,y)]
    \]
    and, if $\barS\neq D_4$, then $w_{\mathcal{C}}\left(\frac{x+\zeta_8'y}{\sqrt{2}},\frac{(\zeta_8')^{-1}x-y}{\sqrt{2}}\right)=\pm w_{\mathcal{C}}(x,y)$.
\end{enumerate}

\end{thmintro}

\begin{remark}
The polynomials that are generators of the polynomial rings in Theorem \ref{thm:reed-muller-sym} are known (see the end of Section \ref{sec:back}).
\end{remark}

A crucial role is played by a divisibility condition proved by Ax in 1964 \cite{A} for the affine version of
Reed-Muller codes. The counterpart for the projective version can be
easily deduced.

Finally, the fourth question is very interesting from a coding theoretical point of view but is left unanswered in this paper whose aim is mainly to study symmetries of weight enumerators. The authors hope to answer it in future work.

In Section \ref{sec:back} we present the necessary background about coding
theory, invariant theory and we recall some properties of Reed-Muller codes.
Section \ref{sec:sym} is devoted to the proof of Theorem \ref{thm:finite-roots}, the discussion of the possible finite groups of symmetries and a classification of codes with weight enumerators having an infinite group of symmetries.
In Section \ref{sec:algo} we present our algorithm for finding the group of symmetries of weight enumerators.
Finally, Section \ref{sec:Reed-Muller} is concerned with the proof of Theorem \ref{thm:reed-muller-sym} and some further results about Reed-Muller codes over larger fields.

\section{Background}\label{sec:back}
For the convenience of the reader we recall here some definitions and results of coding theory and invariant theory which will be useful in the sequel. Three standard references are \cite{MS}, \cite{HP10} and \cite{RS}.

\subsection{Codes and weight enumerators}

A \emph{linear code} $\calC$ is a subspace of $\F_q^n$, where $n$ is
a positive integer called the \emph{length} of the code. A
\emph{generator matrix} of a linear code $\calC$ is a matrix whose
rows generate $\calC$. Elements of $\calC$ are called
\emph{codewords}. The support of a codeword $c\in\calC$, denoted
$\supp(c)$, is defined as
\[
\supp(c):=\{i\in \{1,\ldots,n\} \ | \ c_i\neq 0\}.
\]
The \emph{weight} $\weight(c)$ of a codeword $c$ is the cardinality its support.
The \emph{weight enumerator} $w_\calC(x,y)$ is the polynomial
\[
w_\calC(x,y) = \sum_{c \in \calC} x^{n- \weight(c)} y^{\weight(c)} =
\sum_{i=0}^n A_i x^{n-i} y^{i},  \quad A_i := \#\{c\in \calC \ | \ \weight(c)=i\}.
\]

If $\calC\leq \F_q^m$ and $\calD\leq \F_q^n$ are linear codes with
generator matrices $C$ and $D$ respectively, their \emph{direct sum}
is the vector space $\calC \oplus \calD$ naturally embedded in
$\F_q^{m+n}$, i.e. the code with generator matrix
$\left[\begin{smallmatrix}C & 0 \\ 0 & D\end{smallmatrix}\right]$.
Observe that
$w_{(\calC \oplus \calD)}(x,y)= w_\calC(x,y) \cdot w_\calD(x,y)$.

A \emph{monomial transformation} is a linear transformation of the form $\F_q^n \to \F_q^n, v  \mapsto D P v$,
where $D$ is an $n \times n$ diagonal matrix with non-zero diagonal
entries, and $P$ is an $n \times n$ permutation matrix. Two codes
are said to be \emph{equivalent} if one is the image of the other
under a monomial transformation. Observe that in the case $q=2$, a
monomial transformation is just a permutation. It is easy to see
that two equivalent codes have the same weight enumerator.

\vspace{3mm}

\textbf{Convention:}
Since we are interested in weight enumerators, we will usually identify
codes up to equivalence.
In particular, a generator matrix of a code $\calC$ will mean a generator matrix of some
code equivalent to $\calC$.
\vspace{3mm}

Let $m > 1$ be an integer.
A linear code $\calC$ is \emph{divisible by $m$} if the weight of every codeword of
$\calC$ is divisible by $m$.
A code is \emph{divisible} if it is divisible by some $m>1$.

The \emph{dual} of a linear code $\calC \leq \F_q^n$ (denoted
$\calC^\perp$) is the orthogonal space with respect to the standard
inner product of $\F_q^n$, $\langle x, y \rangle = \sum_{i=1}^n x_i
y_i $ for $x, y \in \F_q^n$, i.e.
\[
\calC^\perp:=\{v\in \F_q^n \ | \ \langle v,c\rangle=0 \ \text{for all }c\in \calC\}.
\]
A linear code is \emph{self-dual} if $\calC^\perp = \calC$.

An important relation between the weight enumerator of a code $\calC$ and its dual $\calC^\perp$ is given
by MacWilliams' Theorem.
\begin{theorem*}[MacWilliams,  \cite{MS}]
  Let $\calC\leq \F_q^n$ be a linear code, and $\calC^\perp$ its dual. Then
  \[
  W_{\calC^\perp}(x,y) = \frac{1}{\#\calC} W_\calC(x + (q-1)y, x-y).
  \]
\end{theorem*}

\subsection{Symmetries and invariant ring}

The group $\GL_2(\C)$ acts naturally on the space $\C[x,y]$ on the right:
for $g=\left[\begin{smallmatrix}a&b\\c&d\end{smallmatrix}\right]\in
\GL_2(\C)$, and $p(x,y)\in \C[x,y]$, we have
\[
p^g(x,y): = (p\circ g) (x,y) = p(ax+by,cx+dy).
\]
We set $S(p(x,y)) := \Stab_{\GL_2(\C)} (p(x,y))$, the group of \emph{symmetries} of $p(x,y)$.

\vspace{3mm}

\textbf{Notation:}
We denote by $\zeta_n$ any $n$-th root of unity and by $\zeta'_n$ any \emph{primitive} $n$-th root of unity.
We call $C_n$ the cyclic group of order $n$, $D_n$ the dihedral group of order $2n$, $V_4$ the Klein four-group and $S_n$ and $A_n$ the symmetric group and the alternating group on $n$ symbols, respectively.

\vspace{3mm}

\begin{example}
A code $\calC$ is divisible by $m$ if and only if
\[
d_{m}:=\left[\begin{smallmatrix}1&0\\0&\zeta'_{m}\end{smallmatrix}\right]\in S(w_\calC(x,y)).
\]
\end{example}

\begin{example}
If $\calC$ is self-dual, then $\#\calC = \#\calC^\perp = q^{n/2}$ and so MacWilliams' Theorem implies that
\[ s_q:=q^{-{1/2}}\left[
\begin{smallmatrix}
   1 & q-1 \\
  1 & -1
\end{smallmatrix}\right]\in S(w_\calC(x,y)) \qquad (\star).
\]

\end{example}

A linear code satisfying $(\star)$ is called \emph{formally self
dual}.

\noindent If $G\leq \GL_2(\C)$ is a group of matrices, the ring
\[
\C[x,y]^G:=\{p(x,y) \in \C[x,y] \ | \ p^g(x,y) = p(x,y)\ \text{for all }g\in G\}.
\]
is called the \emph{invariant ring} of $G$.

\begin{theorem*}[Gleason, \cite{Gl}]
The weight enumerator of a self-dual binary linear code which is
\emph{doubly-even}, i.e. is divisible by $4$, lies in the ring
\[
\C[w_{\hat{\mathcal{H}}_3}(x,y),w_{\mathcal{G}_{24}}(x,y)],
\]
where $w_{\hat{\mathcal{H}}_3}(x,y)$ and $w_{\mathcal{G}_{24}}(x,y)$ are the weight enumerators of the extended Hamming code of length $8$ and of the extended binary Golay code of length $24$ respectively. This polynomial ring is the invariant ring of $G:=\langle d_4,s_2\rangle$.
\end{theorem*}

\subsection{Reed-Muller codes}

The \emph{Reed-Muller code} $\mathcal{RM}_q(r,m)$ on $m$ variables,
of degree $r$ and defined over $\F_q$ is the code
\[
\mathcal{RM}_q(r,m):=\left\{(p(v))_{v\in \F_q^m}\ | \ p\in \F_q[x_1,\ldots,x_m]_r\right\},
\]
where $\F_q[x_1,\ldots,x_m]_r$ is the set of polynomials in $m$
variables with coefficients in $\F_q$ and of degree at most $r$.
The code $\mathcal{RM}_q(r,m)$ encodes all hypersurfaces in
$\mathbb{A}^m(\F_q)$ of degree at most $r$, so determining the weight enumerator of such a code is equivalent to counting
$\F_q$-rational points of hypersurfaces in the affine space.

Similarly we can define the \emph{projective Reed-Muller code}
$\mathcal{PRM}_q(r,m)$ on $m$ variables, of degree $r$ and defined
over $\F_q$ in the following way:
\[
\mathcal{PRM}_q(r,m):=\{(p(v))_{v\in R}\ | \ p\in \F_q[x_0,\ldots,x_m]^h_r\} \cup \{0\} ,
\]
where $\F_q[x_1,\ldots,x_m]^h_r$ is the set of degree $r$
homogeneous polynomials in $m+1$ variables with coefficients in
$\F_q$, and $R$ is a set of representatives in $\F_q^{m+1}$ of all the points of
$\mathbb{P}^m(\F_q)$.
Observe that changing the set of representatives $R$ gives rise to an equivalent code.

If $q=2$ we will usually omit the subscript and write $\RM(r,m)$ for $\RM_2(r,m)$ and $\PRM(r,m)$ for $\PRM_2(r,m)$.

Both Reed-Muller codes and projective Reed-Muller codes are
divisible codes for certain parameters, as a consequence of a theorem by Ax.

\begin{theorem*}[Ax, \cite{A}]
For any integers $r$, $m$ and prime power $q=p^v$, the
Reed-Muller code $\mathcal{RM}_q(r,m)$ is divisible by $q^{\lfloor \frac{m-1}{r}\rfloor}$
and this is the largest power of
the prime $p$ with this property.
\end{theorem*}

It follows from an easy argument relating the projective and affine zeroes of a homogeneous polynomial that the projective Reed-Muller code $\mathcal{PRM}_q(r,m)$ is divisible by $q^{\lfloor \frac{m}{r}\rfloor}$.

Note that in general Reed-Muller codes are not self-dual, but if $q$ is a prime power and $r$, $m$ are integers such that $r < m(q-1)$, then
\[
\mathcal{RM}_q(r,m)^\perp=\mathcal{RM}_q(m(q-1)-r-1,m).
\]

Finally, for some parameters the weight enumerators of Reed-Muller codes are known:
\[
\begin{array}{l|l}
    \calC & w_\calC(x,y) \\\hline
     \RM_q(0,m) &  x^{q^m} + (q-1) y^{q^m} \\
     \RM_q(1,m) & x^{q^m} + q(q^m-1)x^{q^{m-1}}y^{(q-1)q^{m-1}} + (q-1)y^{q^m}
\end{array}
\]
The weight enumerators of $\mathcal{RM}_q(m(q-1)-1,m)=\mathcal{RM}_q(0,m)^\perp$ and
$\mathcal{RM}_q(m(q-1)-2,m)=\mathcal{RM}_q(1,m)^\perp$ can then be derived using MacWilliams' transformations.


\section{Group of Symmetries}\label{sec:sym}
This section is devoted to studying the possible groups of symmetries that a weight enumerator can have.
In the first subsection, we prove Theorem \ref{thm:finite-roots} which allows us to decide if a polynomial has a finite or infinite group of symmetries.

If the group is finite, the classification of finite subgroups of $\PGL_2(\C)$ by Blichfeldt and the study of the second cohomology of these subgroups allow us to classify all possible finite groups of symmetries.
This is treated in the second subsection.

The third subsection handles the case of infinite group of symmetries of weight enumerators and gives an almost complete classification of linear codes with this property.

\subsection{Proof of Theorem \ref{thm:finite-roots}}

For a homogeneous polynomial $p(x,y)\in \C[x,y]$ let $V(p(x,y)) \subseteq \P^1(\C)$ denote the projective variety defined by the vanishing of $p(x,y)$, i.e.
  \[
    V(p(x,y)) = \{(x:y) \in \P^1(\C) \: | \: p(x,y) = 0\}.
  \]

\begin{lemma}\label{lemma-finite}
  Let $p(x,y) \in \C[x,y]$ be a homogeneous polynomial of degree $d$, and let $n=\#V(p(x,y))$. If $n\geq 3$, then
  \[
  \#S(p(x,y)) \leq n!\, d.
  \]
\end{lemma}
\begin{proof}


  Let $G := S(p(x,y))$. Since every $g \in G$ fixes $p(x,y)$, $G$ acts on $V(p(x,y))$.
  Scalar matrices fix $V(p(x,y))$ point-wise, so this action induces an action of $\bar G \subseteq \PGL_2(\C)$,
  where $\bar G$ is the image of $G$ in $\PGL_2(\C)$.
  It is well-known that $\PGL_2(\C)$ acts sharply 3-transitively on $\P^1(\C)$.
  Since $\#V(p(x,y)) \geq 3$, every permutation of $V$ is realized by at most one $\bar g$ in $\PGL_2(\C)$, whence $\bar G$ is finite.

  Any $\bar g \in \bar G$ has exactly $d$ pre-images in $G$: if $g\in G$ is such a pre-image,
  then all the pre-images in $\PGL_2(\C)$ are given by $\lambda g$ for $\lambda\in \C\setminus\{0\}$.
  But
  \[
  p^{(\lambda g)}(x,y)= \lambda^{d} p^g(x,y) = \lambda^{d}p(x,y),
  \]
  so that $\lambda g \in G$ if and only if $\lambda^{d} =1$.

  There are $n!$ permutations of $V(p(x,y))$, so at most $n!$ elements in $\bar G$.
  Thus $\#G \leq n! \, d$.
\end{proof}

Theorem \ref{thm:finite-roots} is now an easy consequence.

\begin{proof}[Proof of Theorem \ref{thm:finite-roots}]
If $V(p(x,y))\geq 3$, then $S(p(x,y))$ is finite by Lemma \ref{lemma-finite}.
If $V(p(x,y))<3$, then $p(x,y)$ is conjugate to $x^ny^m$ for $m,n\in \N\cup \{0\}$, which is easily seen to have an infinite group of symmetries.
\end{proof}

\subsection{The finite case}

If a homogeneous polynomial $p(x,y)\in \C[x,y]$ has finite group of symmetries $S(p(x,y))\subseteq\GL_2(\C)$, its image
$\barS(p(x,y))\subseteq \PGL_2(\C)$ is also finite.
The following theorem gives a classification of all finite subgroups of $\PGL_2(\C)$ up to conjugation.

\begin{theorem}[Blichfeldt \cite{B17}]\label{thm:Bli}
If $H\leq {\rm PGL}_2(\mathbb{C})$ is finite, then $H$ is conjugate
to one of the following groups:
\begin{enumerate}
\item $\left\langle\left[\begin{smallmatrix}1&0\\0&\zeta'_n\end{smallmatrix}\right]\right\rangle\cong
C_n$ for a certain $n\in \N$.
\item
$\left\langle\left[\begin{smallmatrix}1&0\\0&\zeta'_n\end{smallmatrix}\right],
\left[\begin{smallmatrix}0&1\\1&0\end{smallmatrix}\right]\right
\rangle\cong D_n$ for a certain $n\in \N$.
\item
$\left\langle\left[\begin{smallmatrix}1&0\\0&-1\end{smallmatrix}\right],
\left[\begin{smallmatrix}\zeta'_4&\zeta'_4\\1&-1\end{smallmatrix}\right]\right
\rangle\cong A_{4}$.
\item
$\left\langle\left[\begin{smallmatrix}1&0\\0&\zeta'_4\end{smallmatrix}\right],
\left[\begin{smallmatrix}\zeta'_4&\zeta'_4\\1&-1\end{smallmatrix}\right]\right
\rangle\cong S_{4}$.
\item
$\left\langle\left[\begin{smallmatrix}\zeta'_5&0\\0&\zeta_5'^4\end{smallmatrix}\right],
\left[\begin{smallmatrix}0&1\\-1&0\end{smallmatrix}\right],
\frac{1}{10}\left[\begin{smallmatrix}(\sqrt{5} + 5)\zeta'_5 - \sqrt{5} - 5 &(-2\sqrt{5} + 10)\zeta'_5 - 3\sqrt{5} + 5 \\ (-2\sqrt{5} + 10)\zeta'_5 - 3\sqrt{5} + 5
&(-\sqrt{5} - 5)\zeta'_5 + \sqrt{5} - 5\end{smallmatrix}\right] \right\rangle\cong A_{5}$.
\end{enumerate}
\end{theorem}

It follows that $S(p(x,y)) \subseteq \GL_2(\C)$ is (up to conjugation) a central
extension of the groups listed above by the cyclic group $C_d=\langle \left[\begin{smallmatrix}\zeta'_d&0\\0&\zeta'_d\end{smallmatrix}\right]\rangle$, where $d$ is the degree of $p(x,y)$.
These are classified by second cohomology groups, which are known to be the following:

\[
  \begin{array}{c|l}
    G & H^2(G, C_d) \\\hline\hline
    C_n & \Z/ \gcd(d,n) \Z \\\hline
    D_n & \left\{
          \begin{array}{ll}
            \Z/2 \Z \oplus \Z/2 \Z & d \text{ even, } n \text{ even}\\
            \Z/2 \Z  & d \text{ even, } n \text{ odd}\\
            0 & \text{otherwise}
          \end{array}
    \right. \\\hline
    A_4 &  \left\{
          \begin{array}{ll}
            \Z/2 \Z  & 2 \,|\, d, 3 \nmid d\\
            \Z/3 \Z  & 2 \nmid d, 3  \,|\, d\\
            \Z/2\Z \oplus \Z/3 \Z  & 2 \,|\, d, 3 \,|\, d\\
            0 & \text{otherwise}
          \end{array}
    \right. \\\hline
    S_4 &  \left\{
          \begin{array}{ll}
            \Z/2\Z \oplus \Z/2 \Z  & d \text{ even} \\
            0 & \text{otherwise}
          \end{array}
    \right. \\\hline
    A_5 &  \left\{
          \begin{array}{ll}
            \Z/2\Z  & d\text{ even} \\
            0 & \text{otherwise}
          \end{array}
                \right.
  \end{array}
  \]

Not all groups in the list above can actually occur, since we are looking at those which have a faithful representation in ${\rm GL}_2(\mathbb{C})$. This is not the case for example for $A_5$: in this case $d$ has to be even and the trivial extension $A_5\times C_d$ is not possible, so that the only possibility is ${\rm SL}(2,5)\times C_{d/2}$.

\begin{remark}
A natural question that arises is whether or not it is possible to realize the groups of Blichfeldt's Theorem as the group of symmetries of the weight enumerator of a code.
We provide some examples for the first four cases:
\begin{enumerate}
\item \label{it:RCn}If $\calC_n$ denotes the $[n,1,n]_q$ repetition code for $n\geq 3$ and $q > 2$, the direct sum $\calC = \calC_n \oplus \calC_{2n}$ realizes the group $C_n$:
\[
w_\calC(x,y) = x^{3n} + (q-1)x^{2n}y^n + (q-1)x^n y^{2n} + (q-1)^2 y^{3n} \quad \implies \quad \barS(w_\calC(x,y)) \cong C_n.
\]
\item \label{it:RDn}The $[n,1,n]_2$ repetition code for $n\geq 3$ realizes $D_n$:
\[
w_\calC(x,y) = x^n + y^n \quad \implies \quad \barS(w_\calC(x,y)) \cong D_n.
\]
\item \label{it:RA4}The $[12,6,6]_3$ ternary Golay code realizes $A_4$:
\[
w_\calC(x,y) = x^{12} + 264 x^6 y^6 + 440 x^3 y^9 + 24 y^{12}
\quad \implies \quad \barS(w_\calC(x,y)) \cong A_4.
\]
\item \label{it:RS4}The $[8,4,4]$ extended Hamming code realizes $S_4$:
\[
w_\calC(x,y) = x^8 + 14x^4 y^4 + y^8
\quad \implies \quad \barS(w_\calC(x,y)) \cong S_4.
\]
\end{enumerate}
These examples can be deduced by straightforward calculations
$($for part \ref{it:RCn} and \ref{it:RDn}$)$ and using our algorithm in Section \ref{sec:algo} $($for part \ref{it:RA4} and \ref{it:RS4}$)$.
It is however not clear if the group $A_5$ can be realized and how to realize it.
\end{remark}

\subsection{The infinite case}
We proceed to the classification of codes with weight enumerators having infinite group of symmetries.
We will need a lemma.
\begin{lemma} \label{lem:oddweight}
  Let $\calC$ be a code over $\F_q$ with $q \neq 2$ such that all codewords of $\calC$ have even weight.
  Let $x,y \in \calC$ be codewords, with $y$ of weight two and support $\supp(y) = \{i,j\}$.
  Then there exists $\lambda\in \F_q$ such that $(x_i, x_j)=(\lambda y_i, \lambda y_j)$.
\end{lemma}
\begin{proof}
  Suppose that $(x_i, x_j)\neq (\lambda y_i, \lambda y_j)$ for every $\lambda\in\F_q$.
  We show that for suitable $\mu \in \F_q$, the codeword $x + \mu y$ has odd weight.
  Since $(x_i, x_j) \neq 0 (y_i, y_j) = (0,0)$ we may assume without loss of generality that
  $x_i \neq 0$.

  If $x_j = 0$, pick any $\mu \in \F_q\setminus \{0, -x_iy_i^{-1}\}$ (which is non-empty since $q>2$).
  We have $x_i+\mu y_i\neq 0$ and therefore $\supp(x + \mu y)=\supp(x) \ \cup \ \{j\}$ has odd cardinality.
  If $x_j \neq 0$, we can set $\mu = -x_i y_i^{-1}$ and find $\supp(x + \mu y)=\supp(x) \setminus \{i\}$, which
  has again odd cardinality.
\end{proof}

As an immediate consequence we have:
\begin{corollary}\label{cor:disj}
Let $\calC$ be a code over $\F_q$ with $q \neq 2$. Assume all
codewords of $\calC$ have even weight. Let $c_1,\ldots,c_r\in \calC$
of weight $2$ such that $c_i \neq \lambda  c_j$ for any $\lambda \in \F_q$ and $i\neq j$. Then
\[
\supp(c_i)\cap \supp(c_j)=\emptyset
\]
for every $i\neq j$.
\end{corollary}

The first classification result is the following.
\begin{lemma} \label{lem:2roots}
  Let $\calC$ be a linear code of even length $n$ over $\F_q$ with $q \neq 2$.
  Suppose that
  \[
  w_\calC(x,y)=(x^2 + ay^2)^{n/2}, \quad a \in \R\setminus\{0\}.
  \]
  Then $a = q-1$ and
  \[
  \calC \cong \bigoplus_{i=1}^{n/2}\langle(1,1)\rangle_{\F_q}.
  \]
\end{lemma}

\begin{proof}
  If $n=2$ it is clear that $\calC=\langle(1,1)\rangle_{\F_q}$.
  Let $n > 2$. Expanding the above expression, we see that $\calC$ has no codewords of odd weight.
  Moreover, the number of codewords of length $2$ is $a_{n-2} = an/2 \neq 0$.
  Let $r:= a_{n-2}/(q-1)$ and let $c_1,\ldots,c_r$ be a set of
  codewords of weight $2$ such that $c_i \neq \lambda  c_j$ for any $\lambda \in \F_q$ and $i\neq j$. They have disjoint supports by Corollary~\ref{cor:disj}.

  Let $S := \bigcup_i\supp c_i$ and let $\calC_S:= \langle c_1, \dotsc, c_r\rangle_{\F_q}$.
  Every codeword $x \in \calC$ can be written as a sum
  \[
  x = y + z, \quad y, z \in \F_q^n,
  \]
  with $\supp(y) \subseteq S$ and $\supp (z) \cap S = \emptyset$.
  By Lemma~\ref{lem:oddweight}, $y$ is in $\calC_S \subseteq \calC$ and thus so is $z$.
  Consequently, $\calC$ is the direct sum
  \[
  \calC = \calC_S \oplus \calC_{S^c},
  \]
  where $\calC_{S^c}=\{c\in \calC \ | \ \supp(c)\cap S
  =\emptyset\}$.
  This implies that $w_\calC(x,y) = w_{\calC_S}(x,y) \cdot w_{\calC_{S^c}}(x,y)$.

  Now observe that $\calC_S$ is monomially equivalent to the code $\bigoplus_{i=1}^r \langle (1,1) \rangle_{\F_q}$,
  and hence its weight enumerator is $w_{\calC_S}(x,y) = (x^{2} + (q-1)y^2)^r$.
  Therefore, we must have $a = (q-1)$. By induction, $\calC_{S^c}\cong \bigoplus_{i=1}^{n/2-r}\langle(1,1)\rangle_{\F_q}$ so that
  $\calC \cong \bigoplus_{i=1}^{n/2} \langle ( 1,1) \rangle_{\F_q}$, as desired.
\end{proof}

Let us now prove the classification theorem for codes whose weight
enumerator has an infinite group of symmetries. Notice that the weight enumerator $w_\calC(x,y)$ of a code $\calC$ satisfies $\#V(w_\calC(x,y))<3$ if an only if $w_\calC(x,1)$ has at most two distinct roots in $\overline{\Z}$ (the ring of algebraic integers).

\begin{theorem}\label{thm:2roots}
Let $\calC \subseteq \F_q^n$ be a linear code with weight enumerator
$w_\calC(x,1) \in \Z[x]$ having at most two distinct roots in
$\overline{\Z}$. Then only the following possibilities can hold:
\begin{itemize}
\item[(a)] $w_\calC(x,y)=x^n$ and $\calC=\{\underline{0}\}$;
\item[(b)] $w_\calC(x,y)=(x+(q-1)y)^n$ and $\calC= \F_q^n$;
\item[(c)] $n$ is even, $w_\calC(x,y)=(x^2+(q-1)y^2)^{n/2}$ and, if $q\ne 2$,
$\calC\cong \bigoplus_{i=1}^{n/2}\langle(1,1)\rangle$.
\end{itemize}
\end{theorem}

\begin{proof}
  Let $-a$, $-b$ be the roots of $w_\calC(x,1)$ in $\overline{\Z}$, so that
  $w_\calC(x) = (x + a)^r (x+ b)^{n-r}$ for $r \in \N$.
  The number of codewords in $\calC$ of weight one is then $l := ra + (n-r)b$.

  First, assume that $l \neq 0$ and let $m = l/(q-1)$.
  Taking arbitrary linear combinations of the codewords of weight one gives a copy of $\F_q^m$ in $\calC$.
  Therefore, $\calC = \calC_1 \oplus \calC_2$,
  with $\calC_1=\F_q^m$ and $\calC_2:=\{c\in \calC \ | \ \supp(c)\cap \supp(d)=\emptyset \ \forall d\in
  \calC_1\}$. Hence
  \[
  w_\calC(x,y) = (x + (q-1)y)^m \cdot w_{\calC_2}(x,y).
  \]
  Consequently, $-(q-1)$ is a root of $f$; we may assume w.l.o.g. that $a = q-1$.
  We get
  \[
  m = \frac{ra + (n-r)b}{q-1} = r + \frac{(n-r)b}{q-1} \geq r.
  \]
  Hence, either $b=0$ and $r=n$, or $(x + (q-1)y)$ divides $(x+by)^{n-r}$,
  which implies $b=q-1$. Both cases give that
  $w_\calC(x,y) = (x + (q-1)y)^n$.
  But this implies $\#\calC = w_\calC(1,1) = q^n = \#\F_q^n$, whence $\calC = \F_q^n$, as
  desired.

  Now assume $\calC$ has no codewords of weight one, i.e. $l = ra+ (n-r)b = 0$.
  If $a$ is real then so is $b$, and both must be non-negative: since $w_\calC(x,y)$ is non-zero and has
  positive coefficients, $w_\calC(r,1)>0$ for any real $r>0$, so $w_\calC(x,1)$ has only non-positive roots.
  Since $ra = -(n-r)b$, we must have $a = b = 0$ whence $w_\calC(x,1)$ has only one root and $\calC = \{\underline{0}\}$.

  If $a$ is non-real, then $a$ and $b$ are complex conjugate algebraic integers, and we must have
  $r = s$, which is possible only if $n$ is even.
  Consequently,
  \[
  w_\calC(x,y) = (x^2 + \Tr(a)xy + \N(a)y^2)^{n/2}
  \]
  where $\Tr(a) = a + \bar a$ and $\N(a) = a\bar a$.
  The fact that $\calC$ has no codeword of weight one implies that $\Tr(a) = 0$ and hence
  \[
  w_\calC(x,y) = (x^2 + \N(a)y^2)^{n/2}.
  \]

  If $q \neq 2$, Lemma~\ref{lem:2roots} gives the desired conclusion about
  $\calC$.
  If $q = 2$, we must show that $\N(a) = 1$. Since $a$ is an
  algebraic integer, $\N(a)\in \Z$. Since $q=2$, the number of
  codewords of weight $n$ is
  \[
  \N(a)^{n/2}=1
  \]
  so that $\N(a)=\pm 1$. But $\N(a)=-1$ is impossible, since $w_\calC(x,y)$ has non-negative coefficients.
\end{proof}

Note that Theorem \ref{thm:2roots} almost classifies, up to monomial
equivalence, all linear codes having weight enumerator with at most
two distinct roots in $\overline{\Z}$. The case $q=2$ is left
unsolved and seems to be a difficult problem. If $q=2$, the
sum of two codewords of weight $2$ cannot have weight $3$, so
the argument in the proof of Lemma \ref{lem:2roots} does not work.

\begin{question}\label{que:class}
Is it possible to classify binary codes of length $n$ with
weight enumerator $(x^2+y^2)^{n/2}$?
\end{question}

Let $\calC$ and $\calC'$ be two codes with weight enumerator
$(x^2+y^2)^{n/2}$ and $(x^2+y^2)^{n'/2}$ respectively. Then
$\calC\oplus\calC'$ has weight enumerator $(x^2+y^2)^{(n+n')/2}$.
Hence, if we let
\[
\mathcal{M}:=\{\text{binary codes of length} \
n \ \text{and weight enumerator} \ (x^2+y^2)^{n/2} \ | \ n\in
2\N\},
\]
we have that $(\mathcal{M},\oplus)$ is a semigroup; in
order to answer positively to Question \ref{que:class} it thus suffices
to find all irreducible elements in $\mathcal{M}$, which means to
find a minimal set of generators of $(\mathcal{M},\oplus)$.
As usual, we consider elements in $\mathcal{M}$ as classes of codes up
to equivalence.

Clearly, the generator with minimum length is the $[2,1,2]$ code
$\mathcal{X}_1:=\langle(1,1)\rangle$. Furthermore, every element in
$\mathcal{M}$ is formally self-dual and all formally self-dual codes
up to length $16$ are classified in \cite{BH01}. From an analysis of
the tables in the paper we find that, up to length $16$, there are
exactly $4$ other irreducible elements of $\mathcal{M}$, namely the
formally self-dual (but not self-dual) $[6,3,2]$ code
$\mathcal{X}_2$ with generator matrix
\[
\left[\begin{smallmatrix}
1 & 0 & 0 & 1 & 1 & 1 \\
0 & 1 & 0 & 1 & 1 & 1 \\
0 & 0 & 1 & 1 & 1 & 1
\end{smallmatrix}\right]
\]
and three $[14,7,2]$ codes, which we call
$\mathcal{X}_3,\mathcal{X}_4$ and $\mathcal{X}_5$, with generator
matrices $[I|X_3]$,$[I|X_4]$ and $[I|X_5]$ respectively, where
\[
X_3:=\left[\begin{smallmatrix}
1&1&1&0&0&0&0\\
1&1&1&0&0&0&0\\
1&1&1&0&0&0&0\\
1&0&0&1&1&1&1\\
1&0&0&1&1&1&1\\
1&0&0&1&1&1&1\\
1&0&0&0&0&0&0
\end{smallmatrix}\right], \quad
X_4:= \left[\begin{smallmatrix}
1&1&1&1&1&0&0\\
1&1&1&1&1&0&0\\
1&1&1&1&1&0&0\\
1&1&1&1&1&0&0\\
1&1&1&1&0&1&0\\
1&1&1&1&0&1&0\\
1&1&1&1&1&1&1
\end{smallmatrix}\right], \quad
X_5:= \left[\begin{smallmatrix}
1&0&1&0&1&0&0\\
1&0&1&0&1&0&0\\
1&0&1&0&1&0&0\\
1&0&1&0&1&0&0\\
1&1&1&0&1&0&1\\
1&1&1&0&1&0&1\\
1&1&1&1&1&1&1
\end{smallmatrix}\right],
\]
and $I$ is the $7\times 7$ identity matrix. It is not clear how to
construct other generators and it seems already too complex for
a software like {\sc Magma} \cite{Magma}. It is not obvious whether or not there
are infinitely many such generators.

We conclude this section showing a relation between our result and
the Gleason-Pierce Theorem (cf. \cite{Ken94}). Recall
that a code is divisible if there exists an integer $m>1$ such
that the weight of every codeword of $\calC$ is divisible by
$m$.

\begin{theorem}[Gleason-Pierce]
Let $\calC$ be a formally self-dual divisible code. Then
\begin{itemize}
\item $q=2$ and $m\in \{2,4\}$,
\item $q=3$ and $m=3$,
\item $q=4$ and $m=2$,
\item or $q$ arbitrary, $m=2$ and
$w_\calC(x,y)=(x^2+(q-1)y^2)^{n/2}$.
\end{itemize}
\end{theorem}

Hence, Theorem \ref{thm:2roots} implies the following.

\begin{corollary}
For $q>4$, if $\calC$ is a formally self-dual divisible code of
length $2n$, then $\calC$ is equivalent to the direct sum of $n$
copies of $\langle(1,1)\rangle_{\F_q}$.
\end{corollary}

\section{The algorithm}\label{sec:algo}

The proof of Lemma~\ref{lemma-finite} gives an algorithm to find the
stabilizer of every weight enumerator.

Let $\calC$ be a linear code. Suppose that its weight
enumerator $w_\calC(x,y)$ is known and of degree $n$.
\begin{itemize}
\item[1.] Set $G:=\emptyset$.
\item[2.] Calculate $V := \{z_1, \dotsc,z_n\}$ the set of roots of
$w_\calC(x,1)$.
\item[3.] Call $V_3$ the set of all ordered 3-subsets of
$V$.\\
Clearly we have $\#V_3=\frac{1}{6}n^3 - \frac{1}{2}n^2 +
\frac{1}{3}n$.
\item[4.] For every triple $\{w_1,w_2,w_3\}\in V_3$:
\begin{itemize}
\item[4a.] Solve the system $z_ia+b-w_iz_ic-w_i d=0$, $i\in \{1,2,3\}$, where the unknowns are $a,b,c,d$.
It has clearly infinitely many solutions depending on one complex
parameter $\lambda$ (the action of $\PGL_2(\C)$ is sharply
3-transitive, as we said). Call
$\underline{a},\underline{b},\underline{c},\underline{d}$ one
solution.
\item[4b.] If $\{\frac{\underline{a}z_i+\underline{b}}{\underline{c}z_i+\underline{d}} \ | \ z_i\in
V\}=V$, then
\begin{itemize}
\item[4bi.] Let
$A:=\left[\begin{smallmatrix}\underline{a}&\underline{b}\\\underline{c}&\underline{d}\end{smallmatrix}\right]$.
\item[4bii.] Calculate $\lambda:=\frac{w_\calC(\underline{b},\underline{d})}{w_\calC(0,1)}$.
\item[4biii.] Let $G:=G\cup\{\zeta_{n}\lambda^{1/n}A\ | \ \zeta_{n}\in \C \ \text{s.t.} \ \zeta_n^n=1\}$.
\end{itemize}
\end{itemize}
\end{itemize}

Then $G$ is equal to $S(w_\calC(x,y))$.
\vspace{2mm}

This algorithm can be implemented easily in {\sc Magma}, but there
is a problem for Step 2: in $\C$, we do not have access to the
exact roots but only to approximations.
There are two ways to solve this.
The first one is to consider the splitting field of $w_\calC(x,1)$
instead of $\C$.
This gives exact results but is computationally more expensive.
The second one is to use approximations of the roots and control the error
to find an approximated version of the stabilizer.
This is done in the Master thesis of the second author \cite{BFBM}.

Finally, we give a small lemma that allows one to prove the triviality of the group of
symmetries of a given polynomial.
Recall that the cross ratio of four points $(z_1: 1), \dotsc, (z_4 :
1)$ is defined as
\[
[z_1, z_2, z_3, z_4] := \frac{(z_1- z_3)(z_2 - z_4)}{(z_1 - z_4)(z_2
- z_3)}.
\]
Make the symmetric group $S_4$ acts on the cross ratios by permuting
the points, and observe that for any $\sigma \in V_4:= \{\id,
(12)(34), (13)(24), (14) (23)\}$, we have $[z_{\sigma(1)},
z_{\sigma(2)}, z_{\sigma(3)}, z_{\sigma(4)}] = [z_1, z_2, z_3,
z_4]$.

Let $\mathcal{Z}$ be a set of at least four complex points. A
4-tuple of distinct elements $(z_1, z_2, z_3, z_4) \in \mathcal{Z}^4$
will be called \emph{critical} if for any 4-tuple of
distinct elements $(y_1, y_2, y_3, y_4) \in \mathcal{Z}^4$, we have
$[z_1, z_2, z_3, z_4] = [y_1, y_2, y_3, y_4]$ if and only if
$(y_1, y_2, y_3, y_4)= (z_{\sigma(1)}, z_{\sigma(2)}, z_{\sigma(3)}, z_{\sigma(4)})$ for some $\sigma\in V_4$.

\begin{lemma}\label{lem-trivial}
Let $p(x,y)\in \C[x,y]$ be a polynomial with $5$
roots $z_1, z_2, z_3, z_4, z_5$ of $p(x,1)$ such that both $(z_1,
z_2, z_3, z_4)$ and $(z_1, z_2, z_3, z_5)$ are critical. Then
$S(p(x,y))$ is trivial.
\end{lemma}

\begin{proof}
Every $\bar g$ in $\PGL_2(\C)$ sends $\{z_1, z_2, z_3, z_j\}$ to itself, for
$j=4,5$, since it must preserve the cross ratio of these four
points.
  If $\bar g$ sends $z_4$ to $z_j$ for $j \in \{1,2,3\}$, then from the fact that $\bar g$ fixes
  $\{z_1, z_2, z_3, z_5\}$ if follows that some element of this set is also sent to $z_j$,
  contradicting the injectivity of $\bar g$.
  Thus $\bar g z_4 = z_4$.
  But the only permutation of $(z_1, z_2, z_3, z_4)$ which fixes the cross ratio and sends $z_4$ to
  $z_4$ is the identity.
  Hence $\bar g$ fixes four points of $\P^1(\C)$.
  Since the action is sharply 3-transitive, $\bar g = \id$ and the conclusion follows.
\end{proof}

\section{Reed-Muller codes}\label{sec:Reed-Muller}
In this section we study Reed-Muller codes in deeper detail.
The first subsection is focused on the binary Reed-Muller codes
and contains the proof of Theorem \ref{thm:reed-muller-sym} as well as a table listing the groups of symmetries of some small Reed-Muller codes.
The second subsection gives some corresponding tables for Reed-Muller codes over larger fields, some classification results and remarks for further development.

\subsection{Binary Reed-Muller codes}
For the next lemma, let us define
\[
  A := \left\langle
    \left[\begin{smallmatrix}
      1& 0 \\
      0 & \zeta'_4
    \end{smallmatrix}\right],
    \left[\begin{smallmatrix}
      0& 1 \\
      1 & 0
    \end{smallmatrix}\right]
  \right \rangle \subseteq \GL_2(\C), \quad \text{and} \quad
  \bar A = \text{image of }A \text{ in } \PGL_2(\C).
  \]

\begin{lemma}\label{lemma:imp}
  Let $\calC$ be a code with $\bar A \subseteq \barS(w_\calC(x,y))$.
  Then either
  \begin{enumerate}
  \item $\barS(w_\calC(x,y)) = \bar A \cong D_4$ or
  \item \label{itL:self-dual} $\calC$ is formally self-dual or
  \item \label{itL:other}${w_\calC}^v(x,y) = \pm w_\calC(x,y)$, where
    $
      v =
      \frac{1}{\sqrt 2}
      \left[\begin{smallmatrix}
        1& \zeta_8' \\
        \zeta_8'^{-1} & -1
      \end{smallmatrix}\right].
    $
  \end{enumerate}
\end{lemma}
\begin{proof}
  Since $\barS(w_\calC(x,y))$ contains $\bar A \cong D_4$, Blichfeldt's Theorem
  implies that either $\barS(w_\calC(x,y)) \cong D_4$ or $\barS(w_\calC(x,y)) \cong S_4$.
  We assume that $\barS(w_\calC(x,y)) \cong S_4$ and we must show that part \ref{itL:self-dual} or \ref{itL:other} is satisfied.

  By Blichfeldt's Theorem again we see that $\barS(w_\calC(x,y))$ is conjugate to
  $\bar B := \langle \bar A, \bar w \rangle \subseteq \PGL_2(\C)$, where $\bar w$ is the image in
  $\PGL_2(\C)$ of $w :=
    \frac{1}{\sqrt 2}
    \left[\begin{smallmatrix}
      1& 1 \\
      1 & -1
    \end{smallmatrix}\right].
  $
  Let $\bar g \in \PGL_2(\C)$ be such that $\barS(w_\calC(x,y))  = \bar g \bar B \bar g^{-1}$.
  The group $\bar g \bar A \bar g^{-1} \subseteq \barS(w_\calC(x,y))$ is isomorphic to $D_4$, and since all copies
  of $D_4$ in $S_4$ are conjugate we may assume that $\bar g$ normalizes $\bar A$.
  A straightforward calculation then shows that $\bar g$ must be the image in $\PGL_2(\C)$ of a
  matrix of the form
  \[
    g_1 = \left[\begin{smallmatrix}
      \zeta_{16} & 0 \\
      0 & \zeta_{16}^{-1}
    \end{smallmatrix}\right] \quad \text{or} \quad
    g_2 = \left[\begin{smallmatrix}
      0& \zeta_{16} \\
      -\zeta_{16}^{-1} & 0
    \end{smallmatrix}\right].
  \]

  Let $g \in \GL_2(\C)$ be a pre-image of $\bar g$ of this form.
  We have that $\barS(w_\calC(x,y)) = \langle \bar A, \bar w' \rangle$ with $\bar w'$ the image in
  $\PGL_2(\C)$ of $w' := g w g^{-1}$, and it follows that there exists some $\lambda \in \C^\times$
  with $\lambda w' \in S(w_\calC(x,y))$.
  Let $p(x,y) =w_\calC(x,y)$.
  Since $w'^2 = I$, we have $p(x,y) = p^{(\lambda w')^2}(x,y) = p^{\lambda ^2I}(x,y) = \lambda^{2 \deg p(x,y)} p(x,y)$
  whence $\lambda^{2 \deg p(x,y)} = 1$.
  Therefore $p^{w'}(x,y) = (p^{(\lambda w')})^{(\lambda^{-1} I)}(x,y) = \lambda^{- \deg p(x,y)} p(x,y) = \pm  p(x,y)$,
  and so $p(x,y)$ is either invariant or \emph{anti-invariant} under $w'$.

  Let us define the following polynomials:
  \[
    f_1(x,y) = w_{\hat{\mathcal{H}}_3}(x,y), \;
    f_2(x,y) = w_{\mathcal{G}_{24}}(x,y),\;
    f_3(x,y) = x^{12} - 33 x^8 y^4 - 33 x^4 y^8 + y^{12}.
  \]
  Recall that, by Gleason's Theorem, the invariant ring of $B := \langle A, w\rangle$ is the ring $\C[f_1(x,y), f_2(x,y)]$
  and hence the invariant ring of $B' := \langle A, w' \rangle$ is easily seen to be
  $\C[f_1'(x,y), f_2'(x,y)]$ with $f_i' = f_i \circ g$ for $i = 1,2$.
  Moreover, if a polynomial is invariant under $A$ but
  \emph{anti-invariant} under $w$, it must be fixed by $C:= \langle A, wAw \rangle$ whence in the invariant ring of $C$ which is
  $\C[f_1(x,y), f_3(x,y)]$.
  Now it is easy to see that the anti-invariance under $w$ implies that the polynomial is in fact in $f_3(x,y)\cdot \C[f_1(x,y), f_2(x,y)]$.
  This implies finally that a polynomial is invariant under $A$ but anti-invariant under $w'$ if and only if it is in $f_3'(x,y)\cdot\C[f_1'(x,y), f_2'(x,y)]$ with $f_i' = f_i \circ g$ for $i=1,2,3$.

  Now if $\zeta_{16}$ is not primitive, then $f_i \circ g = f_i$ for $i = 1,2,3$.
  Therefore if $p(x,y)$ is invariant (resp. anti-invariant) under $w'$, then it is also
  invariant (resp. anti-invariant) under $w$.
  Now $p^w(x,y) $ is the weight enumerator of the dual code and has positive coefficients, so anti-invariance under $w$ is excluded.
  Moreover, $p^w(x,y) = p(x,y)$ if and only if $\calC$ is formally self-dual, which is
  part \ref{itL:self-dual} of the lemma.

  Finally, if $\zeta_{16}$ is primitive then $w'$ is the $v$ from part \ref{itL:other} of the lemma, and the proof is complete.
\end{proof}

Let us continue with the proof of Theorem \ref{thm:reed-muller-sym}.

\begin{proof}[Proof of Theorem \ref{thm:reed-muller-sym}]
  Part \ref{itB:mleqr} is trivial since these codes are simply the whole space $\F_2^{2^m}$.
  For part \ref{itB:m=2r+1} observe that these Reed-Muller codes are precisely the
  self-dual doubly-even ones, and so the result follows from Gleason's Theorem.
  Now remark that part \ref{itB:r<m<2r+1,jneq2} (resp. part \ref{itB:r<m<2r+1,j=2}) is the dual version of part \ref{itB:m>2r+1,ineq2} (resp. part \ref{itB:m=2r+1,i=2}).
  Hence the result follows by replacing $S$ with its conjugate by $w$, and by replacing the
  polynomials by their composition with $w$ (where $w$ is the Mac-Williams' transformation).

  We are left with proving \ref{itB:m>2r+1,ineq2} and \ref{itB:m=2r+1,i=2}.
  For \ref{itB:m>2r+1,ineq2}, we have $i > 2$ and so the divisibility condition and the symmetry in the variables $x$ and $y$ gives that the weight enumerator is invariant under the matrices
  \[
  d_{2^i} =
  \left[\begin{smallmatrix}
    1 & 0 \\
    0 & \zeta'_{2^i}
  \end{smallmatrix}\right] \quad \text{and} \quad
  t = \left[\begin{smallmatrix}
    0 & 1 \\
    1 & 0
  \end{smallmatrix}\right].
  \]
  Hence the corresponding group in $\PGL_2(\C)$ must contain $\langle d_{2^i}, t \rangle \cong D_{2^i}$.
  We want to show that this group of symmetries cannot be larger.
  Using Blichfeldt's Theorem again, this amounts to showing that this group cannot be
  isomorphic to some $D_k$ containing $\langle d_{2^i}, t \rangle$.
  It is not hard to prove that any such group must contain $d_{k} = \left[\begin{smallmatrix}1 & 0 \\ 0 &\zeta_{k}\end{smallmatrix}\right]$ and therefore
  we simply need to show that our code cannot be divisible by some integer $k > 2^i$
  with $2^i | k$.
  The divisibility condition in Ax's theorem is a strict one, and so it cannot
  be divisible by a larger power of $2$.
  Moreover, the minimum distance of $\RM(r,m)$ being $2^{m-r}$ (see \cite{AK92}) we see that any
  integer by which the code is divisible must be a power of 2.
  Therefore the group $\barS(w_\calC(x,y))$ must be equal to $D_{2^i}$.
  It is now an easy exercise to show that a polynomial is symmetric with powers a multiple of $2^i$
  if and only if it is in the ring
  $\C[x^{2^i} + y^{2^i}, x^{2^{i+1}} + 2 (2^{i+1}-1) x^{2^i}y^{2^i}+ y^{2^{i+1}}]$,
  which is precisely the ring mentioned in the theorem.

  For \ref{itB:m=2r+1,i=2}, the assertion about the invariant ring is the same as in \ref{itB:m>2r+1,ineq2}, and the rest follows from
  Lemma \ref{lemma:imp}, observing that these Reed-Muller codes are not formally self-dual.
  \end{proof}

The following table gives the group of symmetries in $\PGL_2(\C)$ of some small Reed-Muller codes up to conjugation.
It was computed using known facts about the weight enumerators and the algorithm of Section \ref{sec:algo}.
\begin{center}
\begin{table}[h]
\caption{$\barS(w_{\mathcal{RM}_2(r,m)}(x,y))$}
\setlength{\doublerulesep}{-1pt}
\begin{tabular}{|c||c|c|c|c|c|c|c|}
  \hline
 $r\backslash m$
    & 1        & 2        & 3        & 4        & 5        & 6  & 7      \\
  \hline
  \hline
  0 & $\infty$ & {$D_4$}    & {$D_8$}    &
  {$D_{16}$}
   & {$D_{32}$} & {$D_{64}$} & {$D_{128}$}\\
  \hline
  1 & $\infty$ & {$D_4$}    & {$S_4$}    & {$D_{8}$}  & {$D_{16}$} & {$D_{32}$} & {$D_{64}$}\\
  \hline
  2 & $\infty$ & $\infty$ & {$D_8$}    & {$D_{8}$}  & {$S_4$}    & {$D_{4}$} & {$D_{8}$} \\
  \hline
  3 & $\infty$ & $\infty$ & $\infty$ & {$D_{16}$} & {$D_{16}$} & {$D_{4}$} & {$S_4$}  \\
  \hline
  4 & $\infty$ & $\infty$ & $\infty$ & $\infty$ & {$D_{32}$} & {$D_{32}$} & {$D_{8}$} \\
  \hline
  5 & $\infty$ & $\infty$ & $\infty$ & $\infty$ & $\infty$ & {$D_{64}$} & {$D_{64}$}\\
  \hline
  6 & $\infty$ & $\infty$ & $\infty$ & $\infty$ & $\infty$ & $\infty$ & {$D_{128}$} \\
  \hline
  \end{tabular}
\end{table}
\end{center}

\subsection{Further results}

In this subsection we give some results on the group of symmetries of Reed-Muller codes over larger fields.

\begin{theorem}\label{thmsd}
  If $q$ is even, $m$ is odd and $2r=m(q-1)-1$ then
  \[
    \barS(w_{\mathcal{RM}_q(r,m)}(x,y))\cong S_4.
    \]
\end{theorem}

\begin{proof}
Under these hypotheses $\mathcal{RM}_q(r,m)$ is self-dual, so Gleason's theorem holds.
\end{proof}

\begin{theorem}\label{thmd}
If
\begin{itemize}
\item $q\in\{3,4,5\}$ and $m\geq 2r+1$, or
\item $q>5$ and $m\geq r+1$,
\end{itemize}
then $\barS(w_{\mathcal{RM}_q(r,m)}(x,y))$ and $\barS(w_{\mathcal{RM}_q(m(q-1)-r-1,m)}(x,y))$ are either cyclic
or dihedral.
\end{theorem}
\begin{proof}
Under these hypotheses, $\mathcal{RM}_q(r,m)$ is divisible by some $n=q^{\left\lfloor \frac{m-1}{r}\right\rfloor} > 5$ by Ax's Theorem.
Therefore
$C_n\cong \langle\left[\begin{smallmatrix}
1&0\\
0& \zeta'_n
\end{smallmatrix}\right]\rangle
\subseteq \barS(w_{\mathcal{RM}_q(r,m)}(x,y))$.
By Blichfeldt's Theorem,  $\barS(w_{\mathcal{RM}_q(r,m)}(x,y))$ must be cyclic or dihedral. The same holds for $\barS(w_{\mathcal{RM}_q(m(q-1)-r-1,m)}(x,y))$, the group of symmetries of the dual, so that it is conjugate to $\barS(w_{\mathcal{RM}_q(r,m)}(x,y))$ by MacWilliams' transformation.
\end{proof}

Notice that with the same arguments, one can prove a similar result for projective Reed-Muller codes.

\begin{remark}
Theorem \ref{thmsd} and \ref{thmd} are the first steps in the understanding of a general picture for the groups of symmetries of Reed-Muller codes over larger fields. However, some of the arguments we used to prove Theorem \ref{thm:reed-muller-sym} do not work in the general case.

The following two tables give some insight on the shape of the group of symmetries of some Reed-Muller codes over $\F_3$ and $\F_4$. They are obtained using the algorithm presented in Section \ref{sec:algo}. Notice that three cells are left empty, because for those parameters the calculation of the weight enumerator takes too much time.
\begin{center}
\begin{table}[h]
\caption{$\barS(w_{\mathcal{RM}_3(r,m)}(x,y))$}
\setlength{\doublerulesep}{-1pt}
\begin{tabular}{|c||c|c|c|c|}
  \hline
 $r\backslash m$
    & 1                         & 2                         & 3                               & 4                            \\
  \hline
  \hline
  0 & {$D_3$} & {$D_9$} & {$D_{27}$}    & {$D_{81}$} \\
  \hline
  1 & {$D_3$} & {$C_3$}    & {$C_9$}          & {$C_{27}$}    \\
  \hline
  2 & $\infty$                  & {$C_3$}    & {$C_3$}        & {$C_{3}$}   \\
  \hline
  3 & $\infty$                  & {$D_9$} & {$C_3$}        &         \\
  \hline
  4 & $\infty$                  & $\infty$                  & {$C_9$}          &         \\
  \hline
  5 & $\infty$                  & $\infty$                  & {$D_{27}$}    & {$C_{3}$}   \\
  \hline
  6 & $\infty$                  & $\infty$                  & $\infty$                        & {$C_{27}$}    \\
  \hline
  7 & $\infty$                  & $\infty$                  & $\infty$                        & {$D_{81}$} \\
  \hline
  \end{tabular}
\end{table}

\begin{table}[h]
\caption{$\barS(w_{\mathcal{RM}_4(r,m)}(x,y))$}
\setlength{\doublerulesep}{-1pt}
\begin{tabular}{|c||c|c|c|}
  \hline
 $r\backslash m$
    & 1                         & 2                                & 3                               \\
  \hline
  \hline
  0 & {$D_8$} & {$D_{16}$}     & {$D_{64}$}    \\
  \hline
  1 & {$V_4$}    & {$C_4$}         & {$C_{16}$}     \\
  \hline
  2 & {$D_8$} & {$\{{\rm Id}\}$} & {$C_4$}         \\
  \hline
  3 & $\infty$                  & {$\{{\rm Id}\}$} & {$\{{\rm Id}\}$}\\
  \hline
  4 & $\infty$                  & {$C_4$}         &             \\
  \hline
  5 & $\infty$                  & {$D_{16}$}     & {$\{{\rm Id}\}$}\\
  \hline
  6 & $\infty$                  & $\infty$                         & {$C_4$}         \\
  \hline
  7 & $\infty$                  & $\infty$                         & {$C_{16}$}     \\
  \hline
  8 & $\infty$                  & $\infty$                         & {$D_{64}$}    \\
  \hline
  \end{tabular}
\end{table}

\end{center}

A remarkable case is that of $\mathcal{RM}_4(1,1)$:
\[
\barS(w_{\mathcal{RM}_4(1,1)}(x,y))=\left\langle\left[\begin{smallmatrix}3-\sqrt{-15}& 6+2\sqrt{-15} \\-4 & \sqrt{-15}-3\end{smallmatrix}\right],\left[\begin{smallmatrix}1&3 \\ 1&-1 \end{smallmatrix}\right] \right\rangle\cong V_4
\]
and $w_{\mathcal{RM}_4(1,1)}(x,y)\in \C[f_1(x,y),f_2(x,y)]$, where
$f_1(x,y):=2x^2 + (3+\sqrt{-15}) xy + (3-\sqrt{-15}) y^2$ and $f_2(x,y):=53x^4 - 36
x^3y - 18 x^2y^2 + 636xy^3 + 213y^4$. The symmetry $\left[\begin{smallmatrix}3-\sqrt{-15}& 6+2\sqrt{-15} \\-4 & \sqrt{-15}-3\end{smallmatrix}\right]$ does not come from divisibility conditions nor from MacWilliams' identities.

\end{remark}

\section*{Acknowledgements}
The first author was partially supported by PEPS - Jeunes Chercheur-e-s - 2017.
Some of the results of this paper are contained in the Master thesis
\cite{BFBM} of the second author. The authors express their
gratitude to Eva Bayer Fluckiger and Peter Jossen for their support
and the fruitful discussions. Moreover, the authors warmly thank the two anonymous reviewers which helped to improve a previous version of this paper.

\bibliographystyle{alpha}

\bigskip 

\noindent MARTINO BORELLO
\vspace{1mm}\\
{\sc Universit\'e Paris 13, Sorbonne Paris Cit\'e,\\
LAGA, CNRS, UMR 7539,\\
Universit\'e Paris 8,\\
F-93430, Villetaneuse\\
France}
\vspace{1mm}\\
E-mail address:
\href{mailto:martino.borello@univ-paris8.fr}{martino.borello@univ-paris8.fr}

\vspace{3mm}

\noindent OLIVIER MILA
\vspace{1mm}\\
{\sc Universit\"{a}t Bern\\
Mathematisches Institut (MAI)\\
Sidlerstrasse 5\\
CH-3012 Bern\\
Switzerland}
\vspace{1mm}\\
E-mail address:
\href{mailto:olivier.mila@math.unibe.ch}{olivier.mila@math.unibe.ch}

\end{document}